\documentclass[runningheads]{llncs}
\usepackage{amsmath,amssymb}
\usepackage{graphicx,rotating}
\usepackage[font=small]{subfig,caption}
\usepackage{paralist,soul,booktabs}
\usepackage{todonotes}
\usepackage[nocompress]{cite}
\usepackage[pdfpagelabels,colorlinks,allcolors=blue]{hyperref}

\graphicspath{{pics/}}

\newcommand{\df}[1]{{\it #1}}

\newcommand{\Oh}{{\ensuremath{\mathcal{O}}}}
\newcommand{\NP}{{\ensuremath{\mathcal{NP}}}}
\newcommand{\Qh}{{\ensuremath{\mathcal{Q}}}}

\newcommand{\ver}{arxiv}
\newcommand{\arxapp}[2]{\ifthenelse{\equal{\ver}{conf}}{#2}{#1}}

\let\doendproof\endproof
\renewcommand\endproof{~\hfill$\qed$\doendproof}

\newcommand{\myparagraph}[1]{\medskip\noindent\textbf{#1}}
\newcommand{\rephrase}[3]{\medskip\noindent\textbf{#1~#2.}\hspace{0.5ex}\emph{#3}}

\newcounter{casecounter}[section]

\newenvironment{csp}[1][]{\par\smallskip
   \noindent \textbf{Case~P.#1}: }{\smallskip}

\newtheorem{prop}{Proposition}

\begin{document}

\title{Queue Layouts of Planar 3-Trees%
\thanks{This work is supported by the DFG grant Ka812/17-1 and DAAD project~57419183}}

\author{Jawaherul~Md.~Alam\inst{1} \and 
Michael~A.~Bekos\inst{2} \and 
Martin~Gronemann\inst{3} \and
Michael~Kaufmann\inst{2} \and 
Sergey~Pupyrev\inst{1}}
\authorrunning{J.~Md.~Alam, M.~A.~Bekos, M.~Gronemann, M.~Kaufmann, S.~Pupyrev}

\institute{
Dep. of Computer Science, University of Arizona, Tucson, USA
\\\email{\{jawaherul,spupyrev\}@gmail.com}
\and
Institut f{\"u}r Informatik, Universit{\"a}t T{\"u}bingen, T{\"u}bingen, Germany
\\\email{\{bekos,mk\}@informatik.uni-tuebingen.de}
\and
Institut f\"ur Informatik, Universit\"at zu K\"oln, K\"oln, Germany
\\\email{gronemann@informatik.uni-koeln.de}
}

\maketitle

\begin{abstract}
A \emph{queue layout} of a graph $G$ consists of a \emph{linear order} of the vertices of $G$ and a partition of the edges of $G$ into \emph{queues}, so that no two independent edges of the same queue are nested. The \emph{queue number} of~$G$ is the minimum number of queues required by any queue layout~of~$G$.

In this paper, we continue the study of the queue number of planar $3$-trees.
As opposed to general planar graphs, whose queue number is not known to be bounded by a constant, the queue number of planar $3$-trees has been shown to be at most seven. In this work, we improve the upper bound to five. We also show that there exist planar $3$-trees, whose queue number is at least four; this is the first example of a planar graph
with queue number greater than three.
\end{abstract}

\section{Introduction}
\label{sec:introduction}

In a \df{queue} layout~\cite{HR92}, the vertices of a graph are restricted to a line and the edges are drawn at different half-planes delimited by this line, called \df{queues}. The task is to find a linear order of the vertices along the underlying line and a corresponding assignment of the edges of the graph to the queues, so that no two independent edges of the same queues are nested; see Fig.~\ref{fig:example}. Recall that two edges are called \df{independent} if they do not share an endvertex. The \df{queue number} of a graph is the smallest number of queues that are required by any queue layout of the graph. Note that queue layouts form the ``dual'' concept of \df{stack} layouts~\cite{Oll73}, which do not allow two edges of the same stack to cross. 

\begin{figure}[h]
	\centering
	\subfloat[\label{fig:goldner-harary-1}{}]
	{\includegraphics[scale=0.78,page=1]{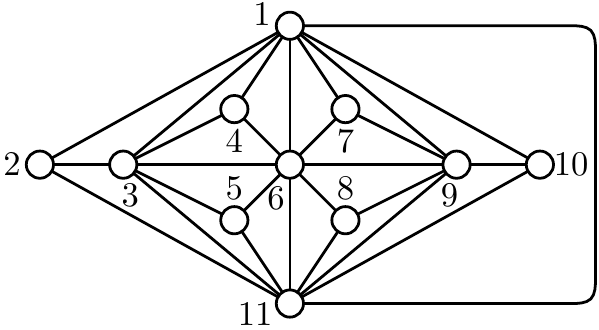}}
	\hfil
	\subfloat[\label{fig:goldner-harary-2}{}]
	{\includegraphics[scale=0.78,page=2]{goldner-harary}}
	\caption{%
	(a)~The Goldner-Harary planar $3$-tree, and	
	(b)~a $5$-queue layout of it produced by our algorithm, in which
	edges of different queues are colored differently.}
	\label{fig:example}
\end{figure}

Apart from the intriguing theoretical interest, queue layouts find applications in several domains~\cite{BCLR96,HLR92,PTT97,T72}. As a result, they have been studied extensively over the years~\cite{T72,RM95,HR92,BFP13,Pem92,DF15,DW05,H03,SS00,Wie17,Pup17}. An important open problem in this area is whether the queue number of {\it planar} graphs is bounded by a constant. A positive answer to this problem would have several important implications, e.g., (i)~that every $n$-vertex planar graph admits a $\Oh(1) \times \Oh(1) \times \Oh(n)$ straight-line grid drawing~\cite{Wood02}, (ii)~that every Hamiltonian bipartite planar graph admits a $2$-layer drawing and an edge-coloring of bounded size, such that edges of the same color do not cross~\cite{DW03}, and (iii)~that the queue number of $k$-planar graphs is also bounded by a constant~\cite{DW05}. The best-known upper bound is due to Dujmovi{\'c}~\cite{Duj15}, who showed that the queue number of an $n$-vertex planar graph is at most $\Oh(\log n)$ (improving upon an earlier bound by Di Battista et al.~\cite{BFP13}). 

It is worth noting that many subclasses of planar graphs have bounded queue number. Every tree has queue number one~\cite{HR92}, outerplanar graphs have queue number at most two~\cite{HLR92}, and series-parallel graphs have queue number at most three~\cite{RM95}.
Surprisingly, planar 3-trees have queue number at most seven~\cite{Wie17}, although they were conjectured  
to have super-constant queue number by Pemmaraju~\cite{Pem92}.
As a matter of fact, every graph that admits a $1$-queue layout is planar with at most $2n-3$ edges; however, testing this property is \NP-complete~\cite{HLR92}; for a survey refer to~\cite{DW05}. 

\smallskip\noindent\textit{Our Contribution.} In Section~\ref{sec:upper-bound}, we improve the upper bound on the queue~number of planar $3$-trees from seven~\cite{Wie17} to five; 
recall that a planar $3$-tree is a triangulated plane graph $G$ with $n \geq 3$ vertices, such that $G$ is either a $3$-cycle, if $n=3$, or has a vertex whose deletion gives a planar $3$-tree with $n-1$ vertices, if $n > 3$.
In Section~\ref{sec:lower-bound}, we show that there exist planar $3$-trees, whose queue number is at least four, thus strengthening a corresponding result of Wiechert~\cite{Wie17} for general (that is, not necessarily planar) $3$-trees. 
We stress that our lower bound is also the best known for planar graphs.
Table~\ref{table:queues} puts our results in the context of existing bounds. 
%
%
We conclude in Section~\ref{sec:conclusions} with open problems.

\begin{table}[!htb]
	\centering
	\caption{Queue numbers of various subclasses of planar graphs}
	\label{table:queues}
	\medskip
	\begin{tabular}{p{2.5cm}|p{2cm}p{1.6cm}|p{1.6cm}p{1.6cm}|}
		\toprule
		& \multicolumn{2}{c|}{Upper bound} & \multicolumn{2}{c|}{Lower bound} \\
		\cmidrule(l{1ex}r{1ex}){2-3}
		\cmidrule(l{1ex}r{1ex}){4-5}
		Graph class & \multicolumn{1}{c}{Old} & \multicolumn{1}{c|}{New} & \multicolumn{1}{c}{Old} & \multicolumn{1}{c|}{New} \\
		\midrule
	    tree  
		& ~$1$~\cite{HR92} & & ~$1$~\cite{HR92} & \\
		
		outerplanar
		& ~$2$~\cite{HLR92} & & ~$2$~\cite{HR92} & \\
		
		series-parallel
		& ~$3$~\cite{RM95} & & ~$3$~\cite{Wie17} & \\
		
		planar 3-tree
		& ~$7$~\cite{Wie17} & $\mathbf{5}$~[Thm.~\ref{thm:upper-bound}] & ~$3$~\cite{Wie17} & $\mathbf{4}$ [Thm.~\ref{thm:lower-bound}] \\
		
		planar
		& ~$\Oh(\log n)$~\cite{Duj15} & & ~$3$~\cite{Wie17} & $\mathbf{4}$ [Thm.~\ref{thm:lower-bound}] \\
		\bottomrule
	\end{tabular}
\end{table}

\paragraph{Preliminaries.}
\label{sec:prel}
For a pair of distinct vertices $u$ and $v$, we write $u \prec v$, if $u$ precedes $v$ in a linear order. We also write $[v_1, v_2, \ldots, v_k ]$ to denote that $v_i$ precedes $v_{i+1}$ for all $1 \leq i < k$. Assume that $F$ is a set of $k \geq 2$ independent edges $(s_i, t_i)$ with $s_i \prec t_i$, for all $1 \leq i \leq k$. If the linear order is $[s_1, \ldots, s_k, t_k, \ldots, t_1]$, then we say that $F$ is a \emph{$k$-rainbow}, while if the linear order is $[s_1, \ldots, s_k, t_1, \ldots, s_k]$,
we say that $F$ is a \emph{$k$-twist}. The edges of $F$ form a \emph{$k$-necklace}, if $[s_1, t_1, \ldots, s_k, t_k]$; see Fig.~\ref{fig:necklace-twist}. A preliminary result for queue layouts is the following.

\begin{lemma}[Heath and Rosenberg~\cite{HR92}]
A linear order of the vertices of a graph admits a $k$-queue layout if and only if there exists no $(k+1)$-rainbow.
\end{lemma}

\noindent Central in our approach is also the following contruction by Dujmovi{\'c} et al.~\cite{DPW04} for internally-triangulated outerplane graphs; for an illustration see Figs.~\ref{fig:G0}--\ref{fig:track-layout}.

\begin{lemma}[Dujmovi{\'c}, Morin, Wood~\cite{DPW04}]
\label{lm:outerplane}	
Every internally-triangulated outerplane graph, $G$, admits a straight-line outerplanar drawing, $\Gamma(G)$, such that the $y$-coordinates of vertices of $G$ are integers, and the absolute value of the difference of the $y$-coordinates of the endvertices of each edge of $G$ is either one or two. Furthermore, the drawing can be used to construct a $2$-queue layout of $G$.
\end{lemma}

Let $\langle u,v,w \rangle$ be a face of a drawing $\Gamma(G)$ produced by the construction of Lemma~\ref{lm:outerplane}, where $G$ is an internally triangulated outerplane graph. Up to renaming of the vertices of this face, we may assume that $|y(u)-y(v)|=|y(u)-y(w)|=1$, $|y(v)-y(w)|=2$ and $y(v) > y(w)$. We refer to vertex $u$ as to the \emph{anchor} of the face $\langle u,v,w \rangle$ of $\Gamma(G)$; $v$ and $w$ are referred to as \emph{top} and \emph{bottom}, respectively. It is easy to verify that drawing $\Gamma(G)$ can be converted to a $2$-queue layout of $G$ as follows:
\begin{inparaenum}[(i)]
\item for any two distinct vertices $u$ and $v$ of $G$, $u \prec v$, if and only if the $y$-coordinate of $u$ is strictly greater than the one of $v$, or the $y$-coordinate of $u$ is equal to the one of $v$, and $u$ is to the left of $v$ in $\Gamma(G)$,
\item edge $(u,v)$ is assigned to the first (second) queue if and only if the absolute value of the difference of the $y$-coordinates of $u$ and $v$ is one (two, respectively) in $\Gamma(G)$.
\end{inparaenum}

Finally, let $\langle u,v,w \rangle$ and $\langle u',v',w' \rangle$ be two faces of $\Gamma(G)$, such that $u$ and $u'$ are their anchors, $v$ and $v'$ are their top vertices, and $w$ and $w'$ are their bottom vertices. If $u$ and $u'$ are distinct and $u \prec u'$ in the $2$-queue layout, then $v \prec v'$ (if $v \neq v'$) and $w \prec w'$ (if $w \neq w'$). The property clearly holds, if $u$ and $u'$ do not have the same $y$-coordinate. Otherwise, the property holds, since $\Gamma(G)$~is~planar.

\begin{figure}[t]
	\centering
	\subfloat[\label{fig:necklace-twist}{}]
	{\includegraphics[scale=0.8,page=1]{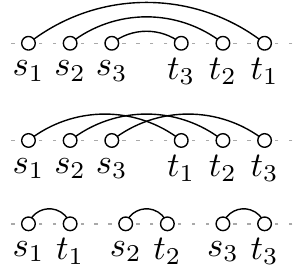}}
	\hfil
	\subfloat[\label{fig:G0}{}]
	{\includegraphics[scale=0.8,page=1]{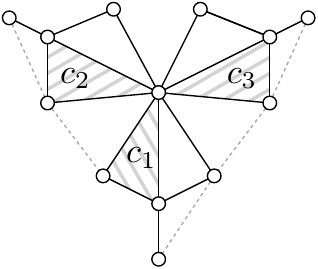}}
	\hfil
	\subfloat[\label{fig:track-layout}{}]
	{\includegraphics[scale=0.8,page=2]{two-levels}}
	\caption{
		(a)~$3$-rainbow, $3$-twist and $3$-necklace (from top to bottom);
		(b)~an internally-triangulated outerplane graph $G_0$;
		the dotted-gray edges are added to make it~biconnected; 
		its gray-shaded faces contain components $c_1$, $c_2$ and $c_3$ of $G_1$;	
		(c)~the drawing $\Gamma(G_0)$ by Lemma~\ref{lm:outerplane}; 
		the vertex-labels indicate the linear order of its $2$-queue layout; 
		the anchor vertices of faces $\langle 9,10,12 \rangle$, $\langle 3,5,9 \rangle$ and $\langle 4,8,9 \rangle$ 
		are $10$, $5$, $8$, respectively.}
	\label{fig:two-levels}
\end{figure}

\section{The Upper Bound}
\label{sec:upper-bound}

In this section, we prove that the queue number of every planar $3$-tree is at most five. Our approach is inspired by the algorithm of Wiechert~\cite{Wie17} to compute $7$-queue layouts for general (not necessarily planar) $3$-trees. To reduce the number of required queues in the produced layouts, we make use of structural properties of the input graph. In particular, we put the main ideas of the algorithm of Wiechert~\cite{Wie17} into a \emph{peeling-into-levels} approach (see, e.g.,~\cite{Yan89}), according to which the vertices and the edges of the input graph are partitioned as follows:
\begin{inparaenum}[(i)]
\item vertices incident to the outerface are at level zero,
\item vertices incident to the outerface of the graph induced by deleting all vertices of levels $0,\ldots,i-1$ are at level $i$,
\item edges between same-level vertices are called \emph{level edges}, and 
\item edges between vertices of different levels are called \emph{binding~edges}.
\end{inparaenum}

To keep the description simple, we first show how to compute a $5$-queue layout of a planar $3$-tree $G$, assuming that~$G$ has only two levels. Then, we~extend our approach to more than two levels. We conclude by discussing the differences between the approach of Wiechert~\cite{Wie17} and ours; we also describe which properties of planar $3$-trees we exploited to reduce the required number of queues.

\myparagraph{The Two-Level Case.}
%
We start with the (intuitively easier) case in which the given planar $3$-tree $G$ consists of two levels, $L_0$ and $L_1$. Since we use this case as a tool to cope with the general case of more than two levels, we consider a slightly more general scenario. In particular, we make the following assumptions (see Fig.~\ref{fig:G0}): 
\begin{inparaenum}[({A.}1)]
\item \label{a:g_0} the graph $G_0$ induced by the vertices of level $L_0$ is outerplane and internally-triangulated, and
\item \label{a:g_1} each connected component of the graph $G_1$ induced by the vertices of level $L_1$ is outerplane and resides within a (triangular) face of $G_0$.
\end{inparaenum}
Without loss of generality we may also assume that $G_0$ is biconnected, as otherwise we can augment it to being biconnected by adding (level-$L_0$) edges without affecting its outerplanarity. Note that in a planar $3$-tree, graph $G_0$ is simply a triangle (and not an outerplane graph, as we have assumed), and as a result $G_1$ is a single outerplane component. Our algorithm maintains the following invariants:
\begin{enumerate}[{I.}1]
\item \label{i:order} the linear order is such that all vertices of $L_0$ precede all vertices of $L_1$; 
\item \label{i:level} the level edges use two queues, $\Qh_0$ and $\Qh_1$; 
\item \label{i:bind} the binding edges use three queues, $\Qh_2$, $\Qh_3$, and $\Qh_4$.

\end{enumerate}
In the following lemma, we show how to determine a (partial) linear order of the vertices of levels $L_0$ and $L_1$ that satisfies the first two invariants of our algorithm.

\begin{lemma}\label{lm:order-level}
There is an order of vertices of level $L_0$ and a partial order of vertices of level $L_1$ such that I.\ref*{i:order} and I.\ref*{i:level} are satisfied.
\end{lemma}
\begin{proof}
To compute an order that satisfies I.\ref*{i:order}, we construct two orders, one for the vertices of level $L_0$ (that satisfies I.\ref*{i:level}) and one for the vertices of level $L_1$ (that also satisfies I.\ref*{i:level}), and then we concatenate them so that the vertices of $L_0$ precede the vertices of $L_1$. 

To compute an order of the vertices of $L_0$ satisfying I.\ref*{i:level}, we apply Lemma~\ref{lm:outerplane}, as by our initial assumption~A.\ref{a:g_0}, graph $G_0$ is internally-triangulated and outerplane. Thus, I.\ref*{i:level} is satisfied for the vertices of level~$L_0$.
To compute an order of the vertices of $L_1$ satisfying I.\ref*{i:level}, we apply Lemma~\ref{lm:outerplane} individually for every connected component of $G_1$, which can be done by our initial assumption~A.\ref{a:g_1}. Then the resulting orders are concatenated (as defined by next Lemma~\ref{lm:bind}). Since for every two connected components of $G_1$, all vertices of the first one either precede or follow all vertices of the second one, we can use the same two queues (denoted by $\Qh_0$ and $\Qh_1$ in I.\ref*{i:level}) for all the vertices of $L_1$. Therefore, I.\ref*{i:level} is satisfied.
\end{proof} 

\noindent Next, we complete the order of the vertices of $G$, in a way that the binding edges between $L_0$ and $L_1$ require at most three additional queues so as to satisfy I.\ref*{i:bind}.

\begin{lemma}\label{lm:bind}
Given the linear order of the vertices of level $L_0$ and the partial order of the vertices of level $L_1$ produced by Lemma~\ref{lm:order-level}, there is a total order of the vertices of $L_0$ and $L_1$ that extends their partial orders and an assignment of the binding edges between $L_0$ and $L_1$ into three queues such that I.\ref*{i:bind} is satisfied.
\end{lemma}

\begin{proof}
Consider a connected component $c$ of $G_1$. By our initial assumption~A.\ref{a:g_1}, component $c$ resides within a triangular face $\langle u,v,w \rangle$ of~$G_0$. Let $u$, $v$ and $w$ be the anchor, top and bottom vertices of the face, respectively. We assign the binding edges incident to $u$ to queue $\Qh_2$, the ones incident to $v$ to queue $\Qh_3$ and the ones incident to $w$ to queue $\Qh_4$; see the blue, red, and green edges in Fig.~\ref{fig:two-levels-layout}. 

\begin{figure}[t]
	\centering
	\includegraphics[scale=0.7,page=3]{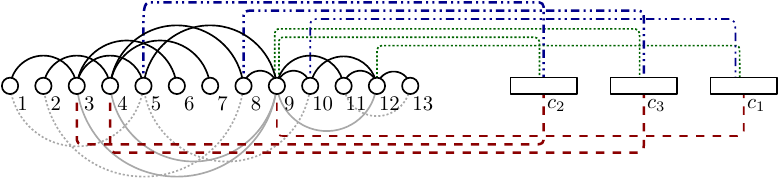}
	\caption{The $5$-queue layout for the graph in Fig.~\ref*{fig:two-levels};
	since $5 \prec 8$ and $8 \prec 10$ in the order of the vertices of level $L_0$ 
	as seen in Fig.~\ref*{fig:two-levels}, 
	$c_2$ precedes $c_3$, and $c_3$ precedes $c_1$.}
	\label{fig:two-levels-layout}
\end{figure}

Next we describe how to compute the relative order of the connected components of $G_1$. Let $c$ and $c'$ be two such components. By our initial assumption~A.\ref{a:g_1}, $c$ and $c'$ reside within two triangular faces $\langle u,v,w \rangle$ and $\langle u',v',w' \rangle$ of $G_0$. Assume that $u$ and $u'$ are the anchors of the two faces, $v, v'$ are top and $w, w'$ are bottom vertices. If $u \neq u'$, then $c$ precedes $c'$ if and only if $u \prec u'$ in the  order of $L_0$. If $u = u'$, we have $v \neq v'$ or $w \neq w'$. If $v \neq v'$, then $c$ precedes $c'$ if and only if $v \prec v'$ in the order of $L_0$. Otherwise (that is, $u = u'$ and $v = v'$), $c$ precedes $c'$ if and only if $w \prec w'$ in the order of $L_0$. We claim that for the resulting order of $L_1$, I.\ref*{i:bind} is satisfied, that is, no two edges of each of $\Qh_2$, $\Qh_3$ and $\Qh_4$ are nested.

We start our proof with $\Qh_2$. Consider two independent edges $(x,y) \in \Qh_2$ and $(x',y') \in \Qh_2$, where $x,x' \in L_0$ and $y,y' \in L_1$ (see the blue edges in Fig.~\ref{fig:two-levels-layout} incident to $5$ and $8$). By construction of $\Qh_2$, $x$ and $x'$ are anchors of two different faces $f_x$ and $f_{x'}$ of $G_0$ (see the faces of Fig.~\ref{fig:track-layout} that contain $c_2$ and $c_3$). Without loss of generality we assume that $x \prec x'$ in the order of $L_0$. Then, the two components $c_y$ and $c_{y'}$ of $G_1$, that reside within $f_x$ and $f_{x'}$ and contain $y$ and $y'$, are such that all vertices of $c_y$ precede all vertices of $c_{y'}$ (in Fig.~\ref{fig:two-levels-layout}, $x=5$ precedes $y=8$; thus, $c_y=c_2$ precedes $c_{y'}=c_3$). Since $y \in c_y$ and $y' \in c_{y'}$, edges $(x,y)$ and $(x',y')$ do not~nest.  

We continue our proof with $\Qh_3$ (the proof for $\Qh_4$ is similar). Let $(x,y)$ and $(x',y')$ be two independent edges of $\Qh_3$, where $x,x' \in L_0$ and $y,y' \in L_1$ (see the red edges in Fig.~\ref{fig:two-levels-layout} incident to $3$ and $4$). By construction of $\Qh_3$, $x$ and $x'$ are the top vertices of two different faces $f_x$ and $f_{x'}$ of $G_0$ (see the faces of Fig.~\ref{fig:track-layout} that contain $c_2$ and $c_3$). Let $c_y$ and $c_{y'}$ be the components of $G_1$ that reside within $f_x$ and $f_{x'}$ and contain $y$ and $y'$. Finally, let $u$ and $u'$ be the anchors of $f_x$ and $f_{x'}$, respectively. Suppose first that $u \neq u'$ and assume that $u \prec u'$ in the order of $L_0$. Since $u \prec u'$, it follows that $x \prec x'$ and that all vertices of $c_y$ precede all vertices of $c_{y'}$ (in Fig.~\ref{fig:two-levels-layout}, $u=5$ precedes $u'=8$, which implies that $x=3$ precedes $x'=4$; thus, $c_y=c_2$ precedes $c_y'=c_3$). Since $y \in c_y$ and $y' \in c_{y'}$, it follows that $(x,y)$ and $(x',y')$ are not nested. Suppose now that $u = u'$ and assume that $x \prec x'$ in the order of $L_0$. Since $u = u'$ and $x \prec x'$, all vertices of $c_y$ precede all vertices of $c_{y'}$. Since $y \in c_y$ and $y' \in c_{y'}$, it follows that $(x,y)$ and $(x',y')$ are not nested. Hence, I.\ref*{i:bind} is satisfied, which concludes the proof.
\end{proof}

Lemmas~\ref{lm:order-level} and~\ref{lm:bind} conclude the two-level case. Before we proceed with the multi-level case, we make a useful observation. To satisfy I.\ref*{i:bind}, we did not impose any restriction on the order of the vertices of each connected component of~$G_1$ (any order that satisfies I.\ref*{i:level} for level $L_1$  would be suitable for us, that is, not necessarily the one constructed by Lemma~\ref{lm:outerplane}). What we fixed, was the relative order of these components. We are now ready to proceed to the multi-level~case.

\myparagraph{The Multi-Level Case.}
%
We now consider the general case, in which our planar $3$-tree $G$ consists of more than two levels, say $L_0,L_1,\ldots,L_\lambda$ with $\lambda \geq 2$. Let $G_i$ be the subgraph of $G$ induced by the vertices of level~$L_i$; $i=0,1,\ldots,\lambda$. The connected components of each graph $G_i$ are internally-triangulated outerplane graphs that are not necessarily biconnected: 
Clearly, this holds for $G_0$, which is a triangle. Assuming that for some $i=1,\ldots,\lambda$, graph $G_{i-1}$ has the claimed property, we observe that each connected component of $G_i$ resides within a facial triangle of $G_{i-1}$. Since each non-empty facial triangle of $G_{i-1}$ in $G$ induces a planar $3$-tree~\cite{MNRA11}, the claim follows by observing that the removal of the outer face of a planar $3$-tree yields a plane graph, whose outer vertices induce an internally-triangulated outerplane graph.

For the recursive step of our algorithm, assume that for some $i=0,\ldots,\lambda-1$ we have a $5$-queue layout for each of the connected components of the graph $H_{i+1}$ induced by the vertices of $L_{i+1},\ldots,L_\lambda$, that satisfies the following invariants.
%
\begin{enumerate}[{M.}1]
\item \label{m:order} the linear order is such that all vertices of $L_j$ precede all vertices of $L_{j+1}$ for every $j=i+1,\ldots,\lambda-1$;
\item \label{m:level} the level edges of $L_{i+1},\ldots,L_\lambda$ use two queues, $\Qh_0$ and $\Qh_1$;
\item \label{m:bind} for every $j=i+1,\ldots,\lambda-1$, the binding edges between $L_j$ and $L_{j+1}$ use three queues, $\Qh_2$, $\Qh_3$, and $\Qh_4$.
\end{enumerate}

Based on these layouts, we show how to construct a $5$-queue layout (satisfying M.\ref*{m:order}--M.\ref*{m:bind}) for each of the connected components of the graph $H_i$ induced by the vertices of $L_i,\ldots,L_\lambda$. Let $C_i$ be such a component. By definition, $C_i$ is delimited by a connected component $c_i$ of $G_i$ which is internally-triangulated and outerplane. If none of the faces of $c_i$ contains a connected component of $H_{i+1}$, then we compute a $2$-queue layout of it using Lemma~\ref{lm:outerplane}. Consider now the more general case, in which some of the faces of $c_i$ contain connected components of $H_{i+1}$. By M.\ref*{m:order}--M.\ref*{m:bind}, we have computed $5$-queue layouts for all the connected components, say $d_1,\ldots,d_k$, of $H_{i+1}$ that reside within the faces of $c_i$. 

\begin{figure}[t]
	\centering
	\subfloat[\label{fig:multi-invariants}{For each of $d_1,\ldots,d_k$ all vertices of $L_j$ precede all vertices of $L_{j+1}$; $j=i+1,\ldots,\lambda-1$}]
	{\includegraphics[width=\textwidth, page=1]{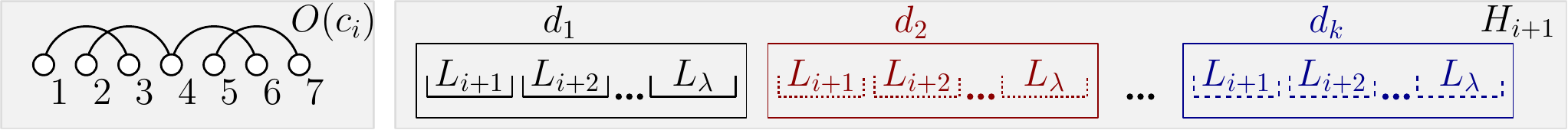}}
	
	\subfloat[\label{fig:multi-order}{The computed linear order based on $p_i,\ldots,p_\lambda$}]
	{\includegraphics[width=\textwidth, page=2]{multi-levels}}
	\caption{Illustrations for the proof of Theorem~\ref{thm:upper-bound}.
	}
	\label{fig:multi-levels}
\end{figure}

We proceed by applying the two-level algorithm to the subgraph of $C_i$ induced by the vertices of $c_i$ and the vertices incident to the outer faces of $d_1,\ldots,d_k$. By the last observation we made in the two-level case, this will result in:
\begin{inparaenum}[(a)]
\item \label{tl:order} a linear order $\mathcal{O}(c_i)$ of the vertices of $c_i$, 
\item \label{tl:comp} a relative order of the components $d_1,\ldots,d_k$,
\item \label{tl:q1q2} an assignment of the (level-$L_i$) edges of $c_i$ into $\Qh_0$ and $\Qh_1$, and
\item \label{tl:q345} an assignment of the binding edges between $c_i$ and each of $d_1,\ldots,d_k$ into $\Qh_2$, $\Qh_3$ and $\Qh_4$.
\end{inparaenum}
Up to renaming, we assume that $d_1,\ldots,d_k$ is the computed order of these components; see Fig.~\ref{fig:multi-invariants}. 

By~(\ref{tl:q1q2}) and~(\ref{tl:q345}), all edges of $C_i$ are assigned to $\Qh_0,\dots,\Qh_4$, since the edges of $d_1,\ldots,d_k$ have been recursively assigned to these queues. 
Next, we partition the order of vertices of $C_i$ into $\lambda-i+1$ disjoint intervals, say $p_i,\ldots,p_\lambda$, such that $p_\mu$ precedes $p_\nu$ if and only if $\mu \prec \nu$. All the (level-$L_i$) vertices of $c_i$ are contained in $p_i$ in the order $\mathcal{O}(c_i)$ by~(\ref{tl:order}). For $j=i+1,\ldots,\lambda$, $p_j$ contains the vertices of $L_j$ of each of the components $d_1,\ldots,d_k$, such that the vertices of  $L_j$ of $d_\mu$ precede the vertices of $L_j$ of $d_\nu$ if and only if $\mu \prec \nu$; see Fig.~\ref{fig:multi-order}. 
The proof that M.\ref*{m:order}--M.\ref*{m:bind} are satisfied can be found in \arxapp{Appendix~\ref{app:upper-bound}}{the full version~\cite{arxiv}}. We summarize in the~following.

\newcommand{\upperbound}{Every planar 3-tree has queue number at most $5$.}
\begin{theorem}\label{thm:upper-bound}
\upperbound
\end{theorem}

We note here that queue layouts are closely related to track layouts; for definitions refer to~\cite{DPW04}. The following result follows immediately from a known result by Dujmovi{\'c}, Morin, Wood~\cite{DMW05}; see \arxapp{Appendix~\ref{app:upper-bound}}{the full version~\cite{arxiv}} for details.

\begin{corollary}\label{thm:sp-track}
The track number of a planar 3-tree is at most $4000$.
\end{corollary}

\noindent\textbf{Differences with Wiechert's algorithm.} Wiechert's algorithm~\cite{Wie17}~builds upon a previous algorithm by Dujmovi{\'c} et al.~\cite{DMW05}. Both  yield queue layouts for general $k$-trees, using the breadth-first search (BFS) starting from an arbitrary vertex $r$ of $G$. For each $d > 0$ and each connected component $C$ induced by the vertices at distance~$d$~from $r$, create a node (called \emph{bag}) ``containing'' all vertices of $C$; two bags are adjacent if there is an edge of~$G$ between them. For a $k$-tree, the result is a tree of bags $T$, called \df{tree-partition}, so that 
\begin{inparaenum}[(P.1)]
\item \label{prp:intra} every node of $T$ induces a connected $(k-1)$-tree, and 
\item \label{prp:inter} for each non-root node $x \in T$, if $y \in T$ is the parent of $x$, then the vertices in $y$ having a neighbor in $x$ form a clique of size $k$.
\end{inparaenum}
Both algorithms order the bags of $T$, such that the vertices of the bags at distance $d$ from $r$ precede those at distance $d+1$. The vertices within each bag are ordered by induction using P.\ref*{prp:intra}. 

The algorithms differ in the way the edges are assigned to queues; the more efficient one by Wiechert~\cite{Wie17} uses $2^k-1$ queues ($2^{k-1}$ for the inter- and $2^{k-1}+1$ for the intra-bag edges), which is worst-case optimal for $1$- and $2$-trees.

If $G$ is a planar $3$-tree and the BFS is started from a dummy vertex incident to the three outervertices of $G$, then the intra- and inter-bag edges correspond to the \df{level} and \df{binding} edges of our approach, while the bags at distance $d$ from $r$ in $T$ correspond to different connected components of level~$d$.

To reduce the number of queues, we observed that in $G$ 
\begin{inparaenum}[(i)]
\item \label{prp:obs2} every node of $T$ induces a connected outerplanar graph, while 
\item \label{prp:obs1} each clique of size three by P.\ref*{prp:inter} is a triangular face of $G$. 
\end{inparaenum}
By the first observation, we reduced the number of queues for intra-bag edges; by the second, we combined orders from different bags more efficiently.

\section{The Lower Bound}
\label{sec:lower-bound}

In the following, we prove that the queue number of planar 3-trees is at least four.
To this end, we will define recursively a subgraph of a planar 3-tree $G$ and we will show that it contains at least one $4$-rainbow in any ordering. 
Starting with a set of $T$ independent edges $(s_i, t_i)$ with $1 \leq i \leq T$ and $T$ to be determined later, we connect their endpoints 
to two unique vertices, say $A$ and $B$, which we assume to be neighboring.
We refer to these edges as \emph{$(s,t)$-edges}.

\begin{figure}[t]
	\centering
	\subfloat[\label{fig:construction_1}{}]
	{\includegraphics[page=4, scale=0.7]{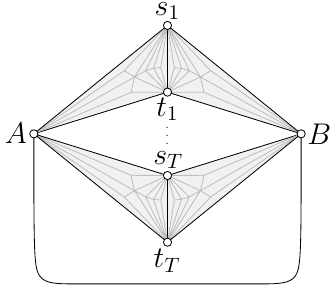}}
	\hfil
	\subfloat[\label{fig:construction_2}{}]
	{\includegraphics[page=2, scale=0.7]{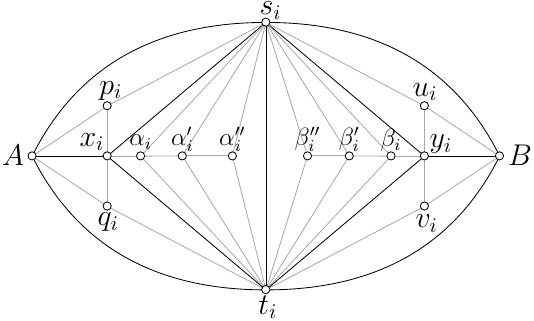}}
	\caption{Construction of graph $G_T$: Each gray subgraph in (a) corresponds to a copy of the graph of (b).}
	\label{fig:construction}
\end{figure}

As a next step, we \emph{stellate} each triangle $\langle A,s_i,t_i \rangle$ with a vertex $x_i$, that is, we introduce vertex $x_i$ and connect it to $A$, $s_i$, and $t_i$.
Symmetrically, we also stellate each triangle $\langle B,s_i,t_i \rangle$ with a vertex $y_i$. 
Afterwards, we add one more level, that is, we stellate each of the triangles 
$\langle A,s_i,t_i \rangle$, 
$\langle B,s_i,t_i \rangle$,
$\langle A,x_i,s_i \rangle$, 
$\langle A,x_i,t_i \rangle$, 
$\langle B,y_i,s_i \rangle$ and
$\langle B,y_i,t_i \rangle$  
with vertices 
$\alpha_i$,
$\beta_i$, 
$p_i$, 
$q_i$, 
$u_i$ and
$v_i$, 
respectively; see Fig.~\ref{fig:construction_2}.
We further stellate $\langle s_i, t_i,\alpha_i\rangle$ with $\alpha'_i$ and then $\langle s_i,t_i,\alpha'_i\rangle$ with $\alpha''_i$. 
Symmetrically, we stellate $\langle s_i, t_i, \beta_i \rangle$ with $\beta'_i$ and $\langle s_i, t_i, \beta'_i \rangle$ with $\beta''_i$. 

Let $G_T$ be the graph constructed so far. 
We refer to vertices $A$ and $B$ as the \emph{poles} of $G_T$ and we assume that $G_T$ admits a $3$-queue layout $\Qh$. 
By symmetry, we may assume that $A \prec B$ and that $s_i \prec t_i$ for each edge $(s_i, t_i)$.
Consider a single edge $(s_i, t_i)$ and the relative order of its endvertices to $A$ and $B$. 
Then, there exist six possible permutations:
\begin{inparaenum}[(P.1)]
\item \label{p0} $s_i \prec A \prec B \prec t_i$, 
\item \label{p1} $A \prec s_i \prec B \prec t_i$,
\item \label{p2} $s_i \prec A \prec t_i \prec B$,
\item \label{p4} $A \prec B \prec s_i \prec t_i$, 
\item \label{p5} $s_i, \prec t_i \prec A \prec B$, and
\item \label{p6} $A \prec s_i \prec t_i \prec B$. 
\end{inparaenum}

By the pigeonhole principle and by setting $T = 6l$, we may claim that at least one of the permutations P.\ref{p0}-P.\ref{p6} applies to at least $l$ edges.
We will show that if too many $(s,t)$-edges share one of the permutations P.\ref{p0}-P.\ref{p5}, then there exists a $4$-rainbow, 
contradicting the fact that $\Qh$ is a $3$-queue layout for $G_T$. 
This implies that if $T$ is large enough, then for at least one $(s,t)$-edge of $G_T$ permutation~P.\ref{p6} applies.
Based on this fact, we describe later how to augment the graph that we have constructed so far
using a recursive construction such that we can also rule out permutation~P.\ref{p6}.
Thereby, proving the claimed lower bound of four.
We start with an auxiliary lemma.

\begin{lemma}
In every queue that contains $r^2$ independent edges, there exists either an $r$-twist or an $r$-necklace.
\end{lemma}
\begin{proof}
Assume that no $r$-twist exists, as otherwise the lemma holds.
We will prove the existence of an $r$-necklace.
Let $(s_1, t_1), \ldots, (s_{r^2}, t_{r^2})$ be the $r^2$ independent edges.
Assume w.l.o.g.\ that $s_i \prec s_{i+1}$ for each $i=1,\ldots,r^2-1$.
Consider the edge $(s_1, t_1)$. 
Since $s_1$ is the first vertex in the order and no two edges nest, 
each vertex $t_i$, with $i>1$, is to the right of $t_1$.
Since no $r$-twist exists, vertex $s_r$ is to the right of $t_1$.
Thus, $(s_1,t_1)$ and $(s_r,t_r)$ do not cross.
The removal of $(s_1,t_1),\ldots,(s_{r-1},t_{r-1})$ makes $s_r$ first. 
By applying this argument $r-1$ times, 
we obtain that $(s_1,t_1), (s_r,t_r), \ldots \left(s_{(r-1)^2+1}, t_{(r-1)^2+1}\right)$ form an $r$-necklace.
\end{proof}
Applying the pigeonhole principle to a $k$-queue layout, we obtain the following.
\begin{corollary}\label{cor:twist-neck}
Every $k$-queue layout with at least $kr^2$ independent edges contains at least one $r$-twist or at least one $r$-necklace.
\end{corollary}

We exploit this result for permutations P.\ref{p0}-P.\ref{p6} as follows. 
Recall that $\Qh$ is a $3$-queue layout for $G_T$. 
So, if we set $T = 18r^2$ for an $r > 0$ of our choice, then at least $3r^2$ $(s,t)$-edges of $G_T$ share the same permutation. 
Moreover, these edges are by construction independent. 
Therefore, by Corollary~\ref{cor:twist-neck} at least $r$ of them form a necklace or a twist (while also sharing the same permutation).
In the following, we show that if $r$ $(s,t)$-edges, say w.l.o.g.~$(s_1,t_1),\ldots,(s_r,t_r)$, form a necklace or a twist (for an appropriate choice of $r$) and simultaneously share one of the permutations P.\ref{p0}-P.\ref{p5}, then a $4$-rainbow is inevitably induced, which contradicts the fact that $\Qh$ is a $3$-queue layout.
We consider each case separately.

\begin{csp}[\ref*{p0}] \label{csp:0}
Let $r=8$.
It suffices to consider the case, in which $(s_1,t_1),\ldots,(s_8,t_8)$ form a twist, since in general for $r>1$ the necklace case is impossible.
Hence, the order is $[ s_1 \ldots s_8 AB t_1 \ldots t_8]$.
We show that $x_4$ always yields a 4-rainbow; 
Fig.~\ref{fig:sabt} shows the three subcases arising when $x_4$ is such that $x_4 \prec B$~holds. 
Clearly, each yields a $4$-rainbow.
Since we did not use the edge $(x_4, A)$, by symmetry, a $4$-rainbow is also obtained when $B \prec x_4$.
\end{csp}

\begin{csp}[\ref*{p1}] \label{csp:1}
As in the previous case, we set $r=8$ and we only consider the case, in which $(s_1,t_1),\ldots,(s_8,t_8)$ form a twist, since the necklace case is again impossible.
Hence, the order is $[A s_1 \ldots s_8 B t_1 \ldots t_8]$.
One may verify that placing $x_4$ and $x_5$ to the left of $t_8$ always results in a 4-rainbow 
(see~\arxapp{Appendix~\ref{app:lower-bound}}{the full version~\cite{arxiv}} for details).
For the case in which $x_4$ and $x_5$ are preceded by $t_8$, we distinguish between if $x_4 \prec x_5$ holds or not.
Both result in a $4$-rainbow.
\end{csp}

\begin{csp}[\ref*{p2}] \label{csp:2}
This case can be ruled out like Case~P.\ref*{p1} due to symmetry.
\end{csp}

\begin{figure}[t]
	\centering
	\subfloat[\label{fig:sabt_1}{$x_4 \prec s_3$}]
	{\includegraphics[page=1, scale=0.8]{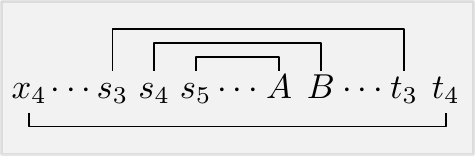}}
	\hfil
	\subfloat[\label{fig:sabt_2}{$s_3 \prec x_4 \prec A$}]
	{\includegraphics[page=2, scale=0.8]{sABt_twist}}
	\hfil
	\subfloat[\label{fig:sabt_3}{$A \prec x_4 \prec B$}]
	{\includegraphics[page=3, scale=0.8]{sABt_twist}}
	\caption{Illustration for the Case~P.\ref*{p0} when $x_4 \prec B$ holds.}
	\label{fig:sabt}
\end{figure}

\begin{csp}[\ref*{p4}] \label{csp:4}
Let $r=10$. We distinguish two subcases based on whether the edges $(s_1,t_1),\ldots,(s_{10},t_{10})$ form a twist or a necklace 
(in contrast to the previous case, here both cases are possible). 

We start with the twist case. 
Hence, the order is $[AB s_1 \ldots s_{10} t_1 \ldots t_{10}]$.
Let $Z_{4\ldots7} = \{ x_4, \ldots, x_7\} \cup \{y_4, \ldots y_7 \}$ and let $z_{4\ldots7}$ be any element of $Z_{4\ldots7}$.
Similar to the previous case, we sweep from left to right and rule out easy subcases. 
However, we have to ensure that we do not use any edge from $z_{4\ldots7}$ to $A$ or $B$ in order
to keep the roles of $x_i$ and $y_i$ interchangeable. 
Fig.~\ref{fig:abst_twist} shows that we may assume that $t_9 \prec z_{4\ldots7} $, that is, 
all $x_4, \ldots, x_7$ and $y_4, \ldots, y_7$ are preceded by $t_9$.

\begin{figure}[t]
	\centering
	\subfloat[\label{fig:abst_twist_1}{$z_{4\ldots7} \prec A$}]
	{\includegraphics[page=1, scale=0.7]{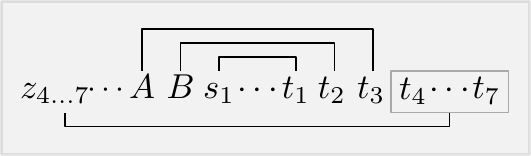}}
	\hfil
	\subfloat[\label{fig:abst_twist_2}{$A \prec z_{4\ldots7} \prec B$}]
	{\includegraphics[page=2, scale=0.7]{ABst_twist_final}}
	\hfil
	\subfloat[\label{fig:abst_twist_3}{$B \prec z_{4\ldots7} \prec s_3$}]
	{\includegraphics[page=3, scale=0.7]{ABst_twist_final}}
	\hfil
	\subfloat[\label{fig:abst_twist_4}{$s_3 \prec z_{4\ldots7} \prec t_3$}]
	{\includegraphics[page=4, scale=0.73]{ABst_twist_final}}
	\hfil
	\subfloat[\label{fig:abst_twist_5}{$t_3 \prec z_{4\ldots7} \prec t_8$}]
	{\includegraphics[page=5, scale=0.73]{ABst_twist_final}}
	\hfil
	\subfloat[\label{fig:abst_twist_6}{$t_8 \prec z_{4\ldots7} \prec t_9$}]
	{\includegraphics[page=6, scale=0.73]{ABst_twist_final}}
	\caption{Illustration for the Case~P.\ref*{p4} when $z_{4\ldots7} \prec t_9$ holds.}
	\label{fig:abst_twist}
\end{figure}

Next, we show that we can always construct a $3$-rainbow spanning $(s_8, t_8)$, which then yields the desired $4$-rainbow. 
Let us take a closer look at the ordering of the 8 vertices in $Z_{4\ldots7}$.
To prevent the creation of a $3$-rainbow that spans $(s_8, t_8)$, 
we claim that the ordering has to comply with two requirements:
\begin{inparaenum}[(R.1)]
\item \label{r1} the indices of the first 7 elements of $Z_{4\ldots7}$ are non-decreasing, and 
\item \label{r2} for the last 7 elements of $Z_{4\ldots7}$, it must hold that all $x$ precede all $y$. 
\end{inparaenum}
Assume to the contrary, that R.\ref{r1} does not hold. 
Hence, there exists a pair of vertices, say w.l.o.g $x_j \prec x_i$, with $i <  j$ and $x_i$ is not the last element of $Z_{4\ldots7}$. 
Then, $[s_i \ldots s_j \ldots x_j \ldots x_i]$ forms a $2$-rainbow and together with the last element of $Z_{4\ldots7}$ that is adjacent to either $A$ or $B$, we obtain a $3$-rainbow spanning $(s_8, t_8)$; a contradiction.
Assume now that R.\ref{r2} does not hold. 
Then, there exists a pair $y_i \prec x_j$ with $y_i$ not being the first element.
Let the first element be $x_l$. 
Then, $[A \ldots B \ldots s_l  \ldots x_l \ldots y_i \ldots x_j]$ is a $3$-rainbow spanning $(s_8, t_8)$; a contradiction.

Now, we show that R.\ref{r1} and R.\ref{r2} cannot simultaneously hold, which implies the existence of a $4$-rainbow. 
Consider the last element of $Z_{4\ldots7}$. 
Assume that~R.\ref{r1} and R.\ref{r2} both hold. 
By~R.\ref{r2}, we may deduce that the last three elements of $Z_{4\ldots7}$ belong to $\{ y_4, \ldots y_7 \}$.
Let them be $y_i,y_j,y_\ell$ as they appear from left to right. 
Then, by~R.\ref{r1} we have that $i<j$. 
Consider now $x_j$. By~R.\ref{r1}, $y_i \prec x_j$ must hold.
This contradicts the fact that $y_i,y_j,y_\ell$ are the last three elements of $Z_{4\ldots7}$.

We continue with the necklace case. Here, the order is $[AB s_1 t_1 \ldots s_{10} t_{10}]$.
We make several observations about the ordering in the form of propositions; 
their formal proofs can be found in \arxapp{Appendix~\ref{app:lower-bound}}{the full version~\cite{arxiv}}.
\begin{prop}
\label{prop:1}
Let $w$ be a neighbor of $s_i$ and $t_i$ for $3\le i\le 8$. 
Then, either $s_{i-1} \prec w \prec t_{i+1}$ holds, or $s_{10}\prec w$.
\end{prop}
\begin{prop}
\label{prop:2}
Let $w$ and $z$ be two vertices that form a $K_4$ with $s_i$ and $t_i$, for $3\le i\le 8$. 
Then, at least one of the following holds: $s_{10}\prec w$ or $s_{10}\prec z$.
\end{prop}
\begin{prop}
\label{prop:3}
Let $w$, $z$ be neighbors of both $s_i$, $t_i$, for $3\le i\le 8$. 
Then, at most one of $w$ and $z$ is between $s_{i-1}$ and $s_i$ or between $t_i$ and $t_{i-1}$. 
Furthermore, if one of $w$ and $z$ is between $s_{i-1}$ and $s_i$ or between $t_i$ and $t_{i-1}$, then the other is not between $s_i$ and $t_i$.
\end{prop}
\begin{prop}
\label{prop:4}
For $4\le i\le 8$, each vertex from the set $\{x_i, y_i, p_i, q_i, u_i, v_i\}$ is between $s_{i-1}$ and $t_{i+1}$.
\end{prop}

\begin{figure}[t]
	\centering
	\subfloat[\label{fig:abst_final_1}{$x_i\prec s_i\prec t_i\prec y_i$}]
	{\includegraphics[page=7, scale=0.8]{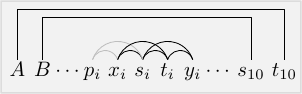}}
	\hfil
	\subfloat[\label{fig:abst_final_2}{$s_i\prec x_i\prec y_i\prec t_i$}]
	{\includegraphics[page=8, scale=0.8]{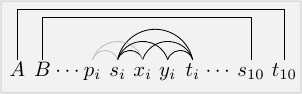}}
\caption{Contradiction for placing $x_i, y_i, p_i, q_i, u_i, v_i$ in range $(s_{i-1}, t_{i+1})$, $4\le i\le 8$.}
\label{fig:abst_final}
\end{figure}

By Proposition~\ref{prop:4}, for $4\le i\le 8$, each vertex from $\{s_i, t_i, x_i, y_i, p_i, q_i, u_i, v_i\}$ is in $(s_{i-1}, t_{i+1})$. Then, the edges between these vertices cannot form a $2$-rainbow, as otherwise this $2$-rainbow along with the two edges $(A, t_{10})$ and $(B, s_{10})$ would form a $4$-rainbow. Assume w.l.o.g.~that $x_i\prec y_i$. Then, by Proposition~\ref{prop:3}, one of the following two conditions hold: (i) $x_i\prec s_i\prec t_i\prec y_i$, (ii) $s_i\prec x_i\prec y_i\prec t_i$; see Fig.~\ref{fig:abst_final}. In both cases, $p_i$ must precede both $x_i$ and $s_i$, as otherwise either $(p_i, s_i), (x_i, t_i)$, or $(p_i, x_i), (s_i, t_i)$ would form a $2$-rainbow; see Fig.~\ref{fig:abst_final}. But then there is no valid position for $q_i$ without creating a $2$-rainbow in either case, resulting together with $(A,t_{10})$ and $(B, s_{10})$ in a $4$-rainbow.
\end{csp}

\begin{csp}[\ref*{p5}] \label{csp:5}
This case can be ruled out like Case~P.\ref*{p4} due to symmetry.
\end{csp}

From the above case analysis it follows that if $r$ is at least $10$ (which implies that $T$ is at least 1,800), then for at least one $(s,t)$-edge of $G_T$ permutation~P.\ref{p6} applies, that is, there exists $1 \leq i_0 \leq T$ such that $A \prec s_{i_0} \prec t_{i_0} \prec B$. Notice that the edges $(A,B)$ and $(s_{i_0},t_{i_0})$ form a $2$-rainbow. 

We proceed by augmenting graph $G_T$ as follows. For each edge $(s_i,t_i)$ of $G_T$, we introduce a new copy of $G_T$, which has $s_i$ and $t_i$ as poles. Let $G'_T$ be the augmented graph and let $(s_1',t_1'),\ldots,(s_T',t_T')$ be the $(s,t)$-edges of the copy of graph $G_T$ in $G'_T$ corresponding to the edge $(s_{i_0},t_{i_0})$ of the original graph $G_T$. Then, by our arguments above there exists  $1 \leq i_0' \leq T$ such that $s_{i_0} \prec s_{i_0}' \prec t_{i_0}' \prec s_{i_0}$. Hence, the edges $(A,B)$, $(s_{i_0}',t_{i_0}')$ and $(s_{i_0},t_{i_0})$ form a $3$-rainbow, since $A \prec s_{i_0} \prec t_{i_0} \prec B$ holds. If we apply the same augmentation procedure to graph $G_T'$, then we guarantee that the resulting graph $G_T''$, which is clearly a subgraph of a planar $3$-tree, has inevitably a $4$-rainbow. 
Hence, either $G_T$ does not admit a $3$-queue layout, as we initially assumed, or $G_T''$ does not admit a $3$-queue layout. In both cases, Theorem~\ref{thm:lower-bound} follows.
%

\begin{theorem}\label{thm:lower-bound}
There exist planar 3-trees that have queue number at least $4$.
\end{theorem}

\section{Conclusions}
\label{sec:conclusions}

In this work, we presented improved bounds on the queue number of planar $3$-trees. Three main open problems arise from our work. The first one concerns the exact upper bound on the queue number of planar $3$-trees. Does there exist a planar $3$-tree, whose queue number is five (as our upper bound) or the queue number of every planar $3$-tree is four (as our lower bound example)? The second problem is whether the technique that we developed for planar $3$-trees can be extended so to improve the upper bound for the queue number of general (that is, non-planar) $k$-trees, which is currently exponential in $k$~\cite{Wie17}. Finally, the third problem is the central question in the area. Is the queue number of general planar graphs (that is, that are not necessarily planar $3$-trees) bounded by a constant?

\bibliographystyle{splncs03}
\bibliography{references}

\arxapp{
\clearpage
\appendix
\section*{\LARGE Appendix}

\section{Omitted Proofs from Section~\ref*{sec:upper-bound}}
\label{app:upper-bound}

\rephrase{Theorem}{\ref*{thm:upper-bound}}{\upperbound}

\begin{proof}
Since M.\ref*{m:order} is clearly satisfied for $C_i$, it remains to prove that the assignment of the edges of $C_i$ to  $\Qh_0,\dots,\Qh_4$ is such that M.\ref*{m:level} and M.\ref*{m:bind} are satisfied.

Since the edges of $H_i$ are partitioned into level and binding, the endvertices of each edge are either in the same or in two consecutive intervals; in the former (latter) case, it is assigned to $\Qh_0$ or $\Qh_1$ (to $\Qh_2$, $\Qh_3$ or $\Qh_4$), since it is a level (binding) edge. Edges assigned to $\Qh_0$ and $\Qh_1$ cannot nest, as otherwise our two-level algorithm has computed an invalid assignment for the level edges of $c_i$ or an invalid assignment in $\Qh_0$ and $\Qh_1$ has been recursively computed for $d_1\ldots,d_k$. Similarly, any two (binding) edges of $\Qh_2$, $\Qh_3$ or $\Qh_4$ cannot be nested, if both bridge $c_i$ with the same component or with two different components of $L_{i+1}$, or both belong to the same component $d_j$, for some $j=1,\ldots,k$. 
It remains to prove that there exist no two nested edges of $\Qh_2$, $\Qh_3$ or $\Qh_4$ that belong to two different components $d_\mu$ and $d_\nu$ and their endvertices are in two consecutive intervals $p_j$ and $p_{j+1}$, where $1 \leq \mu,\nu \leq k$ and $j=i,\ldots,\lambda-1$. The former holds because of the two-level algorithm. The latter holds because all vertices of $d_\mu$ either precede or follow all vertices of $d_\nu$ in both $p_j$ and $p_{j+1}$ (by the choice of the relative order). So, M.\ref*{m:order}--M.\ref*{m:bind} are satisfied and the proof follows.
\end{proof}

\section{Omitted Proofs from Section~\ref*{sec:lower-bound}}
\label{app:lower-bound}

\begin{figure}[b!]
	\centering
	\subfloat[\label{fig:asbt_1}{$x_4 \prec A$}]
	{\includegraphics[page=1, scale=0.75]{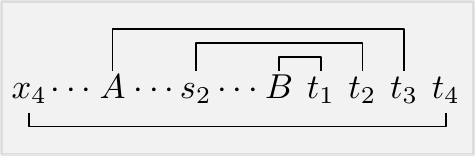}}
	\hfil
	\subfloat[\label{fig:asbt_2}{$A \prec x_4 \prec s_3$}]
	{\includegraphics[page=2, scale=0.75]{AsBt_twist}}
	\hfil
	\subfloat[\label{fig:asbt_3}{$s_3 \prec x_4 \prec B$}]
	{\includegraphics[page=3, scale=0.75]{AsBt_twist}}
	\hfil
	\subfloat[\label{fig:asbt_4}{$B \prec x_4 \prec t_6$}]
	{\includegraphics[page=4, scale=0.75]{AsBt_twist}}
	\hfil
	\subfloat[\label{fig:asbt_5}{$t_6 \prec x_4 \prec t_8$}]
	{\includegraphics[page=5, scale=0.75]{AsBt_twist}}
	\caption{Illustration for the Case~P.\ref*{p1} when $x_4 \prec t_8$ holds.}
	\label{fig:asbt_1_5}
\end{figure}

\par\medskip\noindent\textbf{Details of Case~P.\ref*{p1}}.
Recall that in this case, the order is $[A s_1 \ldots s_8 B t_1 \ldots t_8]$.
Let us consider now every possible position of $x_4$. 
Fig.~\ref{fig:asbt_1_5} shows that a $4$-rainbow is always obtained when $x_4 \prec t_8$. 
Hence, $t_8 \prec x_4$ holds.
Symmetrically, we can obtain that $t_8 \prec x_5$ holds. 
However, we have no knowledge about the relative order of $x_4$ and $x_5$.
In the following, we distinguish two subcases depending on whether $x_4 \prec x_5$ or $x_5 \prec x_4$.
Both subcases are illustrated in Fig.~\ref{fig:asbt_6_7}, which also shows the existence of $4$-rainbows, thus completing this case.

\begin{figure}[t]
	\centering
	\subfloat[\label{fig:asbt_6}{$t_8 \prec x_4 \prec x_5$}]
	{\includegraphics[page=6, scale=0.75]{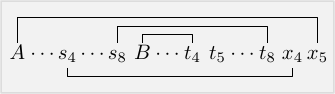}}
	\hfil
	\subfloat[\label{fig:asbt_7}{$t_8 \prec x_5 \prec x_4$}]
	{\includegraphics[page=7, scale=0.75]{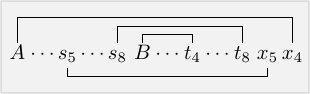}}
	\caption{Illustration for the Case~P.\ref*{p1} when $t_8 \prec x_4$ and $t_8 \prec x_5$ hold.}
	\label{fig:asbt_6_7}
\end{figure}

\medskip

\par\medskip\noindent\textbf{Details of Case~P.\ref*{p4}}. 
In the following, we give the detailed proofs of Propositions~\ref{prop:1}--\ref{prop:4} that we omitted in the main part. 

\rephrase{Proposition}{\ref*{prop:1}}{%
Let $w$ be a neighbor of $s_i$ and $t_i$ for $3\le i\le 8$. 
Then, either $s_{i-1} \prec w \prec t_{i+1}$ holds, or $s_{10}\prec w$.
}
\begin{proof} 
Let $z\in\{s_i, t_i\}$ be the neighbor of $w$. 
We prove in the following that for any placement of $w$, such that neither $s_{i-1}\prec w\prec t_{i+1}$ nor $s_{10}\prec w$ hold, there is a $4$-rainbow:

\begin{figure}[h!]
	\centering
	\subfloat[\label{fig:abst_neck_1}{$w \prec A$}]
	{\includegraphics[page=1, scale=0.75]{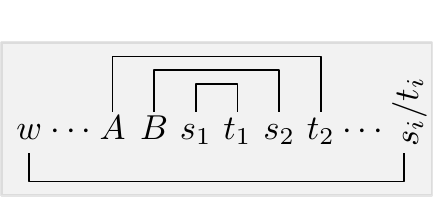}}
	\hfil
	\subfloat[\label{fig:abst_neck_2}{$A \prec w \prec B$}]
	{\includegraphics[page=2, scale=0.75]{ABst_necklace}}
	\hfil\\
	\subfloat[\label{fig:abst_neck_3}{$B\prec w\prec s_{i-1}$}]
	{\includegraphics[page=3, scale=0.75]{ABst_necklace}}
	\hfil
	\subfloat[\label{fig:abst_neck_5}{$t_{i+1}\prec w\prec  s_{10}$}]
	{\includegraphics[page=5, scale=0.75]{ABst_necklace}}
	\caption{Illustrations for the proof of Proposition~\ref{prop:1}.}
	\label{fig:abst_neck}
\end{figure}
\begin{itemize}
	\item[--] if $w\prec A$, then $w\prec A\prec B\prec s_1\prec t_1\prec s_2\prec t_2\prec z$ forms a $4$-rainbow; see Fig.~\ref{fig:abst_neck_1};
	\item[--] if $A\prec w\prec B$, then $A\prec w<B\prec s_1\prec t_1\prec s_2\prec z\prec t_7$ forms a $4$-rainbow; see Fig.~\ref{fig:abst_neck_2};
	\item[--] if $B\prec w\prec s_{i-1}$, then $A\prec B\prec w\prec s_{i-1}\prec t_{i-1}\prec z\prec s_{10}\prec t_{10}$ forms a $4$-rainbow; see Fig.~\ref{fig:abst_neck_3} and
	\item[--] if $t_{i+1}\prec w\prec  s_{10}$, then $A\prec B\prec z\prec s_{T-1}\prec t_{T-1}\prec w\prec s_{10}\prec t_{10}$ forms a $4$-rainbow; see Fig.~\ref{fig:abst_neck_5}.
\end{itemize}
Since each case yields a $4$-rainbow, the proof follows.
\end{proof}

\rephrase{Proposition}{\ref*{prop:2}}{
Let $w$ and $z$ be two vertices that form a $K_4$ with $s_i$ and $t_i$, for $3\le i\le 8$. 
Then, at least one of the following holds: $s_{10}\prec w$ or $s_{10}\prec z$.
}
\begin{proof} 
Since $s_4$, $t_4$, $w$, $z$ form a $K_4$, in any relative ordering of these four vertices, they form a $2$-rainbow. 
By Proposition~\ref{prop:1}, each of $w$ and $z$ is either between $s_{i-1}$ and $t_{i+1}$, or after $s_{10}$. 
But if both of them were between $s_{i-1}$ and $t_{i+1}$, then
the $2$-rainbow by the edges of the $K_4$, along with the two edges $(A, t_{10})$, $(B, s_{10})$ would form a $4$-rainbow; a contradiction.
Hence, $s_{10}\prec w$ or $s_{10}\prec z$ must hold, as desired.
\end{proof}

\rephrase{Proposition}{\ref*{prop:3}}{%
Let $w$, $z$ be neighbors of both $s_i$, $t_i$, for $3\le i\le 8$. 
Then, at most one of $w$ and $z$ is between $s_{i-1}$ and $s_i$ or between $t_i$ and $t_{i-1}$. 
Furthermore, if one of $w$ and $z$ is between $s_{i-1}$ and $s_i$ or between $t_i$ and $t_{i-1}$, then the other is not between $s_i$ and $t_i$.
}
\begin{proof} 
In each of the cases where 
(i) both $w$ and $z$ are between $s_{i-1}$ and $s_i$, 
(ii) both are between $t_i$ and $t_{i-1}$, 
(iii) one is between $s_{i-1}$ and $s_i$, or $t_i$ and $t_{i-1}$, and the other is between $s_i$ and $t_i$, 
the $2$-rainbow formed by the $K_4$ induced by the vertices $w$, $z$, $s_i$ and $t_i$, along with the two edges $(A, t_{10})$ and $(B, s_{10})$ form a $4$-rainbow. 
\end{proof}

\rephrase{Proposition}{\ref*{prop:4}}{
For $4\le i\le 8$, each vertex from the set $\{x_i, y_i, p_i, q_i, u_i, v_i\}$ is between $s_{i-1}$ and $t_{i+1}$.
}
\begin{proof} 
Let $w$ be any vertex from the set $S=\{x_i, y_i, p_i, q_i, u_i, v_i\}$. By Proposition~\ref{prop:1}, it is sufficient to prove that $w$ is not after $s_{10}$. Assume for a contradiction that $s_{10}\prec w$. Observe that for any vertex $w$ from the set $S$, $w$ has an edge (which we call \emph{long}) with exactly one of $A$, $B$, and an edge (which we call \emph{short}) with at least one of $s_i$, $t_i$. On the other hand, consider the four edges $(x_{i-1}, \alpha_{i-1})$, $(\alpha'_{i-1}, \alpha''_{i-1})$, $(y_{i-1}, \beta_{i-1})$ and $(\beta'_{i-1}, \beta''_{i-1})$, each of which creates a $K_4$ with $s_{i-1}$ and $t_{i-1}$. By Proposition~\ref{prop:2}, one vertex from each of these four edges is after $s_{10}$. By the pigeonhole principle, at least two of these vertices are on the same side of $w$. Call them $a$ and $b$, where $a\prec b$. Then, $s_{i-1}\prec t_{i-1}\prec s_{10}\prec t_{10}\prec a\prec b$ form a $3$-rainbow. This together with the long edge (when $w\prec a\prec b$), or the short edge (when $a\prec b\prec w$) form a $4$-rainbow.
\end{proof}

\section{Track Layouts}
\label{sec:tracks}

A \df{track layout} of a graph $G=(V,E)$ is a partition of its vertices into sequences, called \df{tracks}, such that the vertices in each sequence form an independent set and the edges between each two pairs of tracks form a non-crossing set. Formally, let $\{V_i :~1\le i \le t\}$ be a partition of $V$, such that for every edge $(u, v) \in E$, if $u \in V_i$ and $v \in V_j$ then $i \neq j$. Suppose that $<_i$ is a total order of $V_i$. Then, the ordered set $(V_i, <_i)$ is
called a \df{track} and the partition is called a \df{t-track assignment} of $G$.
An \df{X-crossing} in a track assignment consists of two edges $(u, v)$ and $(x, y)$, such that
$u$ and $x$ are on the same track $V_i$, $v$ and $y$ are on a different track $V_j$ with
$u <_i x$ and $y <_j v$. A \df{track layout} is a track assignment with no X-crossings, and the
\df{track number} is the minimum $k$ such that $G$ has a $k$-track layout.


Track and queue layouts are closely related to each other, as shown by Dujmovi{\'c} et al.~\cite{DPW04}. In particular, every $t$-track graph has a $(t-1)$-queue layout, and every $q$-queue graph has track number at most $4q \cdot 4q^{(2q-1)(4q-1)}$. For the case of graphs with bounded tree-width, a better upper bound on the track number is known. For example, trees have track number $3$~\cite{FLW06}, outerplanar graphs have track number $5$~\cite{DPW04}, series-parallel graphs have track number at most $15$~\cite{GGM03}, and planar 3-trees have track number at most $5415$~\cite{GGM03}. Next we improve the upper bound for the track number of planar 3-trees, utilizing the following relation between acyclic chromatic number of a $q$-queue graph and its track number. Recall that a vertex coloring is \df{acyclic} if there is no bichromatic cycle, that is, every cycle receives at least three colors. 

\begin{lemma}[Dujmovi{\'c}, Morin, Wood~\cite{DMW05}]\label{lem:choratic-track}
Every $q$-queue graph with acyclic chromatic number $c$ has track-number at most $c(2q)^{c-1}$.
\end{lemma}

Since the unique $4$-coloring of a planar 3-tree is acyclic~\cite{FRR01}, combining Lemma~\ref{lem:choratic-track} with Theorem~\ref{thm:upper-bound} we obtain the following result.

\rephrase{Theorem}{\ref*{thm:sp-track}}{
The track number of a planar 3-tree is at most $4000$.
}

}{}

\end{document}